\documentclass[11pt,leqno,fleqn]{article}

\usepackage{tikz}
\usetikzlibrary{automata,arrows,shapes,snakes,positioning}
\tikzset{>=latex}
\tikzset{->-/.style={decoration={
  markings,
  mark=at position #1 with {\arrow{>}}},postaction={decorate}}}

\usepackage{latexsym}
\usepackage{amsmath}
\usepackage{amsthm}
\usepackage{amssymb}
\usepackage{amsfonts}
\usepackage{alltt} 
\usepackage{mathabx} 

\usepackage[linesnumbered,ruled,vlined]{algorithm2e}
\let\oldnl\nl 
\newcommand{\nonl}{\renewcommand{\nl}{\let\nl\oldnl}}

\usepackage{amsbsy}
\usepackage[makeroom]{cancel}  

\usepackage{enumitem} 

\usepackage{nameref}
\usepackage[colorlinks=true,
            citecolor=blue,filecolor=blue,linkcolor=blue,urlcolor=blue,  
            linktoc=page,
            pagebackref=true]{hyperref} 
\usepackage{url}      
\usepackage{afterpage}  
\usepackage{float} 
\floatstyle{boxed} 

\usepackage{times}
\usepackage{bm}   
\usepackage{mathrsfs}  
\usepackage{graphicx}

\usepackage{fancybox}

\newcommand{\Hide}[1]{}

\newif
\ifnote
\notetrue

\newif
\ifTR
\TRtrue

\newtheorem{theorem}{Theorem}
\newtheorem{lemma}[theorem]{Lemma}

\newtheorem{proposition}[theorem]{Proposition}

\newtheorem{restxxx}[theorem]{Restriction}

\newtheorem{agreexxx}[theorem]{Agreement}

\newtheorem{termxxx}[theorem]{Terminology}

\newtheorem{notxxx}[theorem]{Notation}

\newtheorem{assumxxx}[theorem]{Assumption}

\newtheorem{convenxxx}[theorem]{Convention}

\newtheorem{exaxxx}[theorem]{Example}

\newtheorem{exexxx}[theorem]{Exercise}

\newtheorem{remxxx}[theorem]{Remark}

\newtheorem{openxxx}[theorem]{Open Problem}
\newtheorem{conjxxx}[theorem]{Conjecture}

\newtheorem{defxxx}[theorem]{Definition}
\newenvironment{definition}[1]{\begin{defxxx}[\emph{#1}]\rm}%
   {\hfill\QED\end{defxxx}}
\newtheorem{defxxxsansQED}[theorem]{Definition}
   {\end{defxxxsansQED}}

\newenvironment{sketch}
{\smallskip\noindent\ignorespaces\textit{Proof Sketch.}}
{\hfill\QED\medskip}
\newtheorem{Prxxx}[theorem]{Proof}
{\end{Prxxx}} 

  {\addtolength{\leftskip}{#1}\addtolength{\rightskip}{#2}}{\par}

\usepackage[margin=0pt,justification=centerlast,font=small,format=hang,%
labelfont=bf,up,textfont=rm,up]{caption}

\newcommand{\Set}[1]{\{ #1 \}}
\newcommand{\SET}[1]{\bigl\{ #1 \bigr\}}

\newcommand{\B}{{\cal B}}

\newcommand{\N}{{\cal N}}

\newcommand{\bigOO}[1]{{\cal O}(#1)} 

\newcommand{\Let}[3]%
    {\textbf{\textsf{let}}\ {#1}\,{#2}\ \textbf{\textsf{in}}\;{#3}\,}
\newcommand{\Try}[3]%
    {\textbf{\textsf{try}}\ {#1} {#2}\ \textbf{\textsf{in}}\;{#3}\;}
\newcommand{\Mix}[3]%
    {\textbf{\textsf{mix}}\ {#1} {#2}\ \textbf{\textsf{in}}\;{#3}\;}
\newcommand{\LET}[3]%
    {\textbf{\textsf{let}}^{\bm{*}}\ {#1} {#2}\ \textbf{\textsf{in}}\;{#3}\;}
\newcommand{\Letrec}[3]%
    {\textbf{\textsf{letrec}}\ {#1} {#2}\ \textbf{\textsf{in}}\;{#3}\;}

\newcommand{\degreeSym}{\mathit{deg}} 
\newcommand{\degr}[2]{{\degreeSym}_{#1}(#2)}

\newcommand{\ie}{\textit{i.e.}}
\newcommand{\eg}{\textit{e.g.}}
\newcommand{\QED}{{\Large $\square$}} 

\newcommand{\nreals}{\mathbb{R}_{+}}

\newcommand{\reals}{\mathbb{R}}

\newcommand{\intervals}[1]{{\cal I}(#1)}

\newcommand{\IndexSym}{\mathit{index}} 
\newcommand{\Index}[1]{{\IndexSym}(#1)} 
\newcommand{\size}[1]{|\,#1\,|}

\newcommand{\set}[1]{\overline{#1}}

\newcommand{\power}[1]{\mathscr{P}(#1)}

\newcommand{\OutF}[1]{\mathsf{OuterFace}(#1)}
\newcommand{\OutPlan}[2]{{{#1}\text{-}\mathsf{outerplanarity}}(#2)}

\newcommand{\CC}{\mathscr{C}}

\newcommand{\undirect}[1]{{\widecheck{#1}}}
\newcommand{\transA}[1]{{#1}^{\star}} 

\newcommand{\vertices}[1]{{\mathbf{V}(#1)}}
\newcommand{\edges}[1]{{\mathbf{E}(#1)}}
\newcommand{\Vertices}[1]{{\mathbf{V}\big(#1\big)}}

\newcommand{\upperB}{\overline{\it c}}
\newcommand{\lowerB}{\underline{\it c}}
\newcommand{\induce}[1]{{#1}^{\#}}

\begin{document}


\setcounter{page}{1}     
\setcounter{tocdepth}{1} 
\ifTR
  \pagenumbering{roman} 
\else
\fi

\title{A Fixed-Parameter Linear-Time Algorithm \\
       for Maximum Flow in Planar Flow Networks} 
\author{Assaf Kfoury%
           \thanks{Partially supported by NSF awards CCF-0820138
           and CNS-1135722.} \\
        Boston University \\
        \ifTR Boston, Massachusetts \\ 
        \href{mailto:kfoury@bu.edu}{kfoury{@}bu.edu}
        \else \fi
        \Hide{
          \and   Benjamin Sisson%
           \footnotemark[1] \\
        Boston University   \\
        \ifTR Boston, Massachusetts \\ 
        \href{mailto:bmsisson@bu.edu}{bmsisson{@}bu.edu}
        \else \fi
        }
}

\ifTR
   \date{\today}
\else
   \date{} %
\fi
\maketitle
  \ifTR
     \thispagestyle{empty} 
  \else
  \fi

\ifTR
    \tableofcontents
    \newpage
\else
    \vspace{-.2in}
\fi

  \begin{abstract}
  \addcontentsline{toc}{section}{Abstract}

\noindent
We pull together previously established graph-theoretical results to
produce the algorithm in the paper's title. The glue are three easy
elementary lemmas.


  \end{abstract}

  \newpage
  \pagenumbering{arabic}  

\section{Introduction}
\label{sect:intro}

We combine a previous result on what is called \emph{graph
reassembling}, together with a previous result on what are
called \emph{network typings}, in order to show the existence of an
algorithm that returns the value of a maximum flow in planar flow
networks in fixed-parameter linear-time.  Those results are made to
work together by means of three easy elementary lemmas. In this
introductory section we informally explain the notions involved;
formal definitions are in later sections of the report.

One way of understanding the \emph{reassembling} of a simple
undirected graph $G$ is this: It is the process of cutting every edge
of $G$ in two halves, and then splicing the two halves of every edge,
one by one in some order, in order to recover the original $G$. We
thus start from one-vertex components, with one component for each
vertex $v$ and each with $\degr{}{v}$ dangling half edges,%
\footnote{$\degr{}{v}$ is the degree of vertex $v$, the number of
  edges incident to $v$.
  }
and then gradually reassemble larger and larger components of
the original $G$ until $G$ is fully reassembled. One optimization
associated with graph reassembling is to keep the number of dangling
half edges of each reassembled component as small as possible.
Graph reassembling and associated optimization problems are
examined in earlier reports on network
analysis~\cite{Kfoury:SCP2014,SouleBestKfouryLapets:eoolt11,%
kfoury+mirzaei:2017,kfoury+mirzaei:2017B}.

As for \emph{network typings}, these are algebraic or arithmetic
formulations of interface conditions that network components must
satisfy for them to be safely and correctly interconnected. A
particular use of network typings is to quantify desirable properties
related to resource management (\eg, percentage ranges of channel
utilization, mean delays between routers, etc., as well as flow
conservation and capacity constraints along channels), and to enforce
them as invariant properties across network interfaces. More on this
use of network typings is in several
reports~\cite{BestKfoury:dsl11,Kfoury:sblp11,Kfoury:SCP2014}.
In this paper, a typing for a network component $\N$ is limited to
specify a range of admissible values for every combination of input
ports (or ``sources'') and output ports (or ``sinks'') of $\N$.

The parameter to be bounded in the algorithm of our main
result is called the \emph{edge-outerplanarity} of a planar
graph. Edge-outerplanarity is distinct but closely related to the
usual notion of outerplanarity, and was introduced in earlier studies
for other purposes (\eg, disjoint paths in sparse graphs, as
in~\cite{bentz2009}). As with outerplanarity, for a fixed
edge-outerplanarity $k$, the number $n$ of vertices in a graph can be
arbitrarily large. Our main result can be re-phrased thus: For the class
${\CC}_k$ of planar flow networks whose edge-outerplanarity is bounded
by a fixed $k\geqslant 1$, there is an algorithm which, given an arbitrary
$\N\in{\CC}_k$, computes the value of a maximum flow in $\N$ in time
$\bigOO{n}$ where $n = \size{\N}$.

\paragraph{Organization of the Report.}

Section~\ref{sect:preliminaries} is background material that makes
precise many of the notions we use throughout the report.
Section~\ref{sect:transformation} includes the three elementary lemmas
(Lemmas~\ref{lem:basic}, \ref{lem:basic-for-plane},
and~\ref{lem:equivalent-networks}) that we need to pull together the
results on graph reassembling and network typings.

A formal definition of \emph{graph reassembling} -- different from, but
equivalent to, the informal definition above -- 
is in Section~\ref{sect:two-previous}, which includes the
optimization result (Theorem~\ref{thm:about-reassembling})
that we need for the main result.
A formal definition of \emph{network typings} -- also more general
than the informal definition above -- 
is in Section~\ref{sect:two-previous}, where we present the
relevant result about typings (Theorem~\ref{thm:about-typing})
that we use in this paper.

Our main result (Theorem~\ref{thm:our-result}) 
is in Section~\ref{sect:our-result}.
We conclude with a brief discussion of follow-up work in
Section~\ref{sect:future}.


\section{Preliminary Notions}
\label{sect:preliminaries}

 In this paper we need to consider both directed and undirected
 graphs.  We use the same letter `$G$', possibly decorated, to refer
 to both directed and undirected graphs; the context will make clear
 whether $G$ is directed or undirected. We refer to the vertices and
 edges of a graph $G$ by writing $\vertices{G}$ and $\edges{G}$.

 \paragraph{Directed Graphs and Undirected Graphs.}
 Throughout, our undirected graphs are \emph{simple graphs}, \ie,
 they have no self-loops and no multi-edges.
 In particular, an edge is uniquely identified by the two-element set of
 its endpoints $\Set{v,w}$, which we also write as $\set{v\,w}$.
 \Hide{
 We think of the
 two-element set $\Set{v,w}$ as a multiset, so that $\Set{v,v} \neq
 \Set{v}$ and $\set{v\,v}$ is a (undirected) self-loop from vertex $v$
 back to vertex $v$.  (The usual definition of \emph{simple graphs}
 excludes both self-loops and multi-edges; thus our undirected graphs
 are not simple, though they exclude multi-edges.)
 }

 In the case of directed graphs also, we disallow self-loops as well as
 multi-edges with the same direction. However, we allow two edges with
 opposite directions between the same two vertices $v$ and $w$,
 written as the ordered pairs $(v,w)$ and $(w,v)$. We also write
 $\set{v\,w}$ and $\set{w\,v}$ for $(v,w)$ and $(w,v)$, respectively.

 The context will make clear whether $\set{v\,w}$ is an undirected
 edge in an undirected graph, or a directed edge in a directed graph.
 If $\set{v\,w}$ is undirected, then $\set{v\,w} = \set{w\,v}$;
 if $\set{v\,w}$ is directed, then $\set{v\,w} \neq \set{w\,v}$.

 Let $G$ be a directed graph. The undirected version of $G$, denoted
 $\undirect{G}$, consists in ignoring all edge directions. In the
 graphical representation of $G$, all the edges are reproduced in
 $\undirect{G}$, with every arrow
 `$\xrightarrow{\hspace{.8cm}}$' replaced by a line segment
 `\raisebox{.25em}{\rule{.8cm}{.03em}}', \emph{with one exception}: Two
 directed edges between the same two vertices,
 `$v \xrightarrow{\hspace{.8cm}} w$' and `$v \xleftarrow{\hspace{.8cm}} w$',
 are collapsed into a single line segment
 `$v$ \raisebox{.25em}{\rule{.8cm}{.03em}} $w$'.%
 \footnote{By this reasoning and contrary to what is often done elsewhere,
   we do not consider here an undirected graph as a special case of a
   directed graph, whereby every undirected edge $\Set{v,w}$ is viewed as
   being two directed edges $(v,w)$ and $(w,v)$.}

 If $G$ is a directed graph containing
 two edges with opposite directions between the same two
 vertices $\Set{v,w}$, say $e_1 = \set{v_1\,v_2}$ and $e_2 = \set{v_2\,v_1}$,
 then $\Set{e_1,e_2}$ form what we call a \emph{two-edge cycle} in $G$.
 Two-edge cycles do not occur in undirected graphs.

 For a vertex $v$ in a directed graph, we write
 $\degr{\text{in}}{v}$ and $\degr{\text{out}}{v}$ for the in-degree and
 out-degree of $v$. And we write $\degr{}{v}$ for
 $\degr{\text{in}}{v} + \degr{\text{out}}{v}$, the total number of edges
 incident to $v$, both incoming and outgoing.

 \paragraph{Flow Networks.}

 A flow network is a quadruple $(G,c,s,t)$ where $G$ is a directed
 graph, $c:\edges{G}\to\nreals$ is the capacity function on edges, and
 $s$ (the \emph{source}) and $t$ (the \emph{sink}) are two distinct
 members of $\vertices{G}$. Trivially, for the max flow problem from
 $s$ to $t$, there is no loss of generality in assuming that the
 underlying graph $G$ is \emph{connected} and contains
 no \emph{self-loops}.%
         \footnote{We write $\nreals$ for the set of non-negative
         real numbers and $\reals$ for the set of all real numbers.}

 If the underlying graph $G$ of the network is connected, then so is
 its undirected version $\undirect{G}$. Biconnectedness is a stronger
 requirement than connectedness (``there are at least two distinct
 directed paths between any two points'') which we cannot impose on
 $G$.
 
 Nonetheless, we can further assume that, if $G$ is the underlying
 graph of a flow network $(G,c,s,t)$, then $\undirect{G}$ (though not
 $G$ itself) is \emph{biconnected}.  This means there are no cut
 vertices in $\undirect{G}$. Indeed, suppose $\undirect{G}$ is
 connected but not biconnected. If the source $s$ and the sink $t$ are
 in the same component (\ie, maximal biconnected subgraph)
 $\undirect{G}'$ of $\undirect{G}$, we can discard all biconnected
 subgraphs other than $\undirect{G}'$, and compute a max flow from $s$
 to $t$ relative to $G'$ only, where $G'$ is the subgraph of $G$ whose
 undirected version is $\undirect{G}'$. If the source $s$ and the sink
 $t$ are in two distinct components $\undirect{G}'$ and
 $\undirect{G}''$ of $\undirect{G}$, respectively, then there are at
 least $p\geqslant 1$ cut vertices, say $\Set{v_1,\ldots,v_p}$, such
 that all directed paths from $s$ to $t$ in $G$ visit the same $p$
 vertices. For simplicity, suppose $p=1$ and there is only one cut
 vertex $v$ on the directed paths from $s$ and $t$; the argument
 extends to an arbitrary number $p\geqslant 1$ in the obvious way.
 With one cut vertex $v$, we compute a first max flow $f_1$ from $s$
 to $v$ and a second max flow $f_2$ from $v$ to $t$; the max flow in
 the original $G$ is $\max \Set{f_1,f_2}$.

 To compute a max flow in $(G,c,s,t)$ by first identifying the
 biconnected components in the underlying $\undirect{G}$ in a
 preprocessing phase, as suggested in the preceding paragraph, does
 not add more than linear sequential time $\bigOO{m+n}$ or logarithmic
 parallel time $\bigOO{\log n}$ to the overall cost, where
 $m = \size{\edges{G}}$ and $n = \size{\vertices{G}}$; \eg,
 see~\cite{tarjan1984,hochbaum1993,nikolopoulos2007}.

 If $\undirect{G}$ is biconnected, there are no
 vertices $v\in\vertices{G}$ such that $\degr{}{v} = 1$. However,
 there may exist vertices $v\in\vertices{G}$ such that $\degr{}{v} = 2$.
 Consider a fixed $v\in\vertices{G} - \Set{s,t}$ such that $\degr{}{v} = 2$,
 which must therefore occur in the graphical representation of $G$ in one
 of three configurations $\Set{\text{(a)},\text{(b)},\text{(c)}}$ where:
 \[
   \text{(a)}\quad
   v_1 \xrightarrow{\ \ e_1\ \ } v \xrightarrow{\ \ e_2\ \ } v_2,
   \qquad\text{(b)}\quad
   v_1 \xleftarrow{\ \ e_1\ \ } v \xrightarrow{\ \ e_2\ \ } v_2,
   \qquad\text{(c)}\quad
   v_1 \xrightarrow{\ \ e_1\ \ } v \xleftarrow{\ \ e_2\ \ } v_2,  
\]
  for some $v_1,v_2\in\vertices{G}-\Set{v}$. We will assume configurations
  (b) and (c) do not occur in $G$, as they do not contribute any value
  to the max flow from $s$ to $t$.%
  \footnote{We do not suggest that we can allow the presence of
  configurations (b) and (c) in $G$, and then eliminate them in a
  preprocessing phase in linear time. To do the latter in full
  generality, without restrictions on the topology of $G$, would take
  more than $\bigOO{n}$ time though not more than $\bigOO{n^2}$, but
  that would be enough to spoil the linear time of our final result.}
  As for configuration (a), we can delete
  the two edges $\set{v_1\,v}$ and $\set{v\,v_2}$, replace them by a
  single new edge $\set{v_1\,v_2}$, and define the new capacity
  $c(\set{v_1\,v_2}) := \min \Set{c(\set{v_1\,v}), c(\set{v\,v_2})}$;
  clearly, this is can be done without affecting the final value of the
  max flow from $s$ to $t$, and can be done in time $\bigOO{n}$ in
  a preprocessing phase.

  We do not exclude the possibility $\degr{}{s} = 2$ and/or
  $\degr{}{t} = 2$, but in constant time $\bigOO{1}$ we can sligtly
  modify the underlying $G$ to $G'$, and update the capacity function
  $c$ to $c'$, so that $(G,c,s,t)$ is equivalent to $(G',c',s,t)$ and
  $\degr{}{s} = \degr{}{t} = 3$. For example, if $\degr{}{s} = 2$,
  we can do the following: Introduce $3$ fresh vertices $\Set{v_1,v_2,v_3}$
  and three fresh edges $\Set{\set{v_1\,s},\set{v_1\,v_2},\set{v_1\,v_3} }$,
  with $v_2$ and $v_3$ inserted in the two edges incident to $s$,
  and then set $c'(e) = c(e)$ for every edge $e\in\edges{G}$
  and $c'(\set{v_1\,s}) = c'(\set{v_1\,v_2}) = c'(\set{v_1\,v_3}) = 0$.

  Based on the preceding comments, there is no loss of generality in
  making the following assumption ${(\diamondsuit)}$, which is to be
  satisfied by the underlying graph $G$ of every flow network in this paper.
  
\paragraph{Assumption $\bm{(\diamondsuit)}$.}\label{assumption}
  If $G$ is a directed graph, then it satisfies three conditions:
  \begin{itemize}[itemsep=1pt,parsep=2pt,topsep=2pt,partopsep=0pt] 
  \item[(1)] 
  \quad $G$ has no \emph{self-loops},
  \item[(2)] 
  \quad $\degr{}{v}\geqslant 3$ for every $v\in\vertices{G}$, and
  \item[(3)] 
   \quad the undirected version $\undirect{G}$ of $G$ is \emph{biconnected}.
  \end{itemize}
  Note that~\nameref{assumption} does not preclude the presence of
  two-edge cycles in $G$.

\paragraph{Edge Outerplanarity of Plane Graphs.}

  A commonly used parameter of undirected plane graphs is \emph{outerplanarity}.
  A less common parameter is \emph{edge outerplanarity},
  which is also only defined for undirected plane graphs.  We here
  extend both notions to all graphs, directed and undirected.

  We make a distinction between \emph{planar} graphs and \emph{plane}
  graphs.  $G$ is a \emph{plane graph} if it is drawn on the plane
  without any edge crossings. $G$ is a \emph{planar graph} if it is
  isomorphic to a plane graph; \ie, it is embeddable in the plane in
  such a way that its edges intersect only at their endpoints. To keep
  the distinction between the two notions, we define
  the \emph{outerplanarity index} of a \emph{planar} graph and
  the \emph{outerplanarity} of a \emph{plane} graph.

  If $G$ is a plane graph, directed or undirected, then
  the \emph{outerplanarity} of $G$ is the number $k$ of times that all
  the vertices on the outer face (together with all their incident
  edges) have to be removed in order to obtain the empty graph. In
  such a case, we say that the plane graph $G$
  is \emph{$k$-outerplanar}.

  If $G$ is a planar graph, directed or undirected, then
  the \emph{outerplanarity index} of $G$ is the minimum of the
  outerplanarities of all the plane embeddings $G'$ of $G$.

  Deciding whether an arbitrary graph is planar can be carried out in
  linear time $\bigOO{n}$ and, if it is planar, a plane embedding of it
  can also be carried out in linear time~\cite{patrignani2013}.  Given a
  planar graph $G$, the outerplanarity index $k$ of $G$ and a
  $k$-outerplanar embedding of $G$ in the plane can be computed in time
  $\bigOO{n^2}$, and a $4$-approximation of its outerplanarity index can
  be computed in linear time~\cite{Kammer2007}.

We give a formal definition of \emph{edge outerplanarity}, less
common than standard outerplanarity, now also extended to directed graphs.

\begin{definition}{Edge-Outerplanarity}
\label{def:edge-outerplanarity}
  Let $G$ be a plane graph, directed or undirected.
  If $\edges{G} =\varnothing$ and $G$ is a graph of isolated vertices,
  the \emph{edge outerplanarity} of $G$ is $0$. If $\edges{G} \neq\varnothing$,
  we pose $G_0 := G$ and define $K_0$ as the set of edges lying on $\OutF{G_0}$.

  For every $i>0$, we define $G_i$ as the plane graph obtained
  after deleting all the edges in $K_0 \cup \cdots \cup K_{i-1}$ from the
  initial $G$ and $K_i$ the set of edges lying on $\OutF{G_i}$.
  
  The \emph{edge outerplanarity} of $G$, denoted $\OutPlan{E}{G}$, 
  is the least integer $k$ such that $G_{k}$ is a graph without edges,
  \ie, the edge outerplanarity of $G_{k}$ is $0$. This process of peeling
  off the edges lying on the outer face $k$ times produces a $k$-block
  partition of $\edges{G}$, namely, $\Set{K_0,\ldots,K_{k-1}}$.%
   \footnote{There is an unessential difference between our definition here and
     the definition in~\cite{bentz2009}. In Section 2.2 of that
     reference, ``a $k$-edge-outerplanar graph is a planar graph
     having an embedding with \emph{at most} $k$ layers of edges.'' In
     our presentation, we limit the definition to plane graphs and say ``a
     $k$-edge-outerplanar plane graph has \emph{exactly} $k$ layers of
     edges.'' Our version simplifies a few things later.}
\end{definition}

  To keep \emph{outerplanarity} and \emph{edge outerplanarity} clearly
  apart, we call the first \emph{vertex outerplanarity}, or more
  simply \emph{V-outerplanarity}, and the second \emph{edge
  outerplanarity}, or more simply \emph{E-outerplanarity}.

  There is a close relationship between \emph{V-outerplanarity}
  and \emph{E-outerplanarity} (Theorem 4 in Section 5.1
  in~\cite{bentz2009}).  In the case of three-regular plane graphs, the
  relationship is much easier to state. This is
  Proposition~\ref{prop:V-outer-vs-E-outer} next.
  
\Hide{  \emph{V-outerplanarity} and \emph{E-outerplanarity} are ``almost the same'',
  (Proposition~\ref{prop:V-outer-vs-E-outer} below).

We conclude by stating the relationship between the
standard notion of outerplanarity and the notion of
edge-outerplanarity used in this report.  If $G$ is a plane graph, let
$\OutPlan{V\!\!}{G}$ denote the smallest $k$ such that $G$ is
$k$-outerplanar (this is $k$-outerplanarity in the standard sense). }

\begin{proposition}
  \label{prop:V-outer-vs-E-outer}
  If $G$ is a $3$-regular plane graph, directed or undirected, then:
  \[
     \OutPlan{V\!\!}{G} \leqslant \OutPlan{E}{G} 
     \leqslant 1+ \OutPlan{V\!\!}{G}.
  \]
  Thus, for $3$-regular plane graphs,
  \emph{V-outerplanarity} and \emph{E-outerplanarity} are ``almost the same''.
\end{proposition}

\begin{sketch}
  For a $3$-regular plane graph, the difference between
  $\OutPlan{V\!\!}{G}$ and $\OutPlan{E}{G}$ occurs in the last stage
  in the process of repeatedly removing (in the case of standard
  \emph{V-outerplanarity}) all vertices on the outer face and all their incident
  edges. The corresponding last stage in the case
  of \emph{E-outerplanarity} may or may not delete all edges; if
  it does not, then one extra stage is needed to delete all remaining
  edges.
\end{sketch}

The preceding result is not
true for arbitrary plane graphs, even if they are regular.
Consider, for example, the four-regular
plane graph $G$ in Figure~\ref{fig:four-regular-plane-graph}, where
$\OutPlan{V\!\!}{G} = 2$ while $\OutPlan{E}{G} = 4$. 

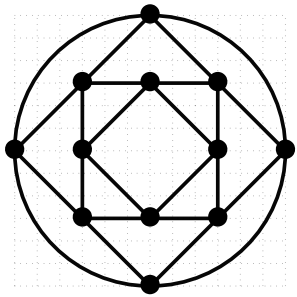
\begin{figure}[H] 
\begin{centering}

\noindent
%
\begin{tikzpicture}[scale=.3] 
       \newcommand\EdgeOpacity{[line width=1.4,black,opacity=.98]};
       \draw [help lines, dotted] (0, 0) grid (12,12);
       \coordinate (A) at (0,0);
       \coordinate (B) at (6,0);
       \coordinate (C) at (12,0);
       \coordinate (D) at (12,6);
       \coordinate (E) at (12,12);
       \coordinate (F) at (6,12);
       \coordinate (G) at (0,12);
       \coordinate (H) at (0,6);
       \coordinate (I) at (3,3);
       \coordinate (J) at (6,3);       
       \coordinate (K) at (9,3);
       \coordinate (L) at (9,6);       
       \coordinate (M) at (9,9);
       \coordinate (N) at (6,9);       
       \coordinate (O) at (3,9);
       \coordinate (P) at (3,6);       

       \draw \EdgeOpacity (B) to[out=0,in=-90] (D) ;
       \draw \EdgeOpacity (D) to[out=90,in=0] (F) ;       
       \draw \EdgeOpacity (B) -- (K) ;              
       \draw \EdgeOpacity (D) -- (M) ;       
       \draw \EdgeOpacity (F) to[out=180,in=90] (H) ;
       \draw \EdgeOpacity (F) -- (O) ;       
       \draw \EdgeOpacity (H) to[out=-90,in=180] (B) ;
       \draw \EdgeOpacity (H) -- (I) ;
       \draw \EdgeOpacity (I) -- (J) ;
       \draw \EdgeOpacity (I) -- (B) ;       
       \draw \EdgeOpacity (J) -- (K) ;
       \draw \EdgeOpacity (J) -- (L) ;       
       \draw \EdgeOpacity (K) -- (D) ;
       \draw \EdgeOpacity (K) -- (L) ;       
       \draw \EdgeOpacity (L) -- (M) ;
       \draw \EdgeOpacity (L) -- (N) ;       
       \draw \EdgeOpacity (M) -- (F) ;
       \draw \EdgeOpacity (M) -- (N) ;       
       \draw \EdgeOpacity (N) -- (O) ;
       \draw \EdgeOpacity (N) -- (P) ;       
       \draw \EdgeOpacity (O) -- (H) ;
       \draw \EdgeOpacity (O) -- (P) ;              
       \draw \EdgeOpacity (P) -- (I) ;       
       \draw \EdgeOpacity (P) -- (J) ;

       \node at (B) {\huge $\bullet$}; 
       \node at (D) {\huge $\bullet$};
       \node at (F) {\huge $\bullet$};   
       \node at (H) {\huge $\bullet$};
       \node at (I) {\huge $\bullet$};
       \node at (J) {\huge $\bullet$};
       \node at (K) {\huge $\bullet$};   
       \node at (L) {\huge $\bullet$}; 
       \node at (M) {\huge $\bullet$};
       \node at (N) {\huge $\bullet$};
       \node at (O) {\huge $\bullet$};
       \node at (P) {\huge $\bullet$};   
\end{tikzpicture}

\end{centering}   
   \caption{A four-regular plane graph $G$, 
   with
   $\OutPlan{V\!\!}{G} = 2$ and $\OutPlan{E}{G} = 4$.
   }
\label{fig:four-regular-plane-graph}
\end{figure}


\section{A Flow-Preserving and Planarity-Preserving Transformation}
\label{sect:transformation}

We define a transformation which, given an arbitrary directed graph $G$
satisfying~\nameref{assumption} on page~\pageref{assumption},
returns a directed graph $\transA{G}$ where:
\begin{itemize}[itemsep=1pt,parsep=2pt,topsep=2pt,partopsep=0pt] 
  \item $\degr{}{v} = 3$ for every vertex $v\in\Vertices{\transA{G}}$, and
  \item there are no two-edge cycles,
\end{itemize}
where $\degr{}{v} = \degr{\text{in}}{v} + \degr{\text{out}}{v}$,
the total number of edges incident to vertex $v$, both incoming and outgoing.
The transformation $G\mapsto \transA{G}$ is defined in terms of an operation
which we call \textbf{expand}.

\begin{definition}{Expand}
\label{def:expand}
 The operation \textbf{expand} is applied to vertices of degrees
 $\geqslant 3$.  Given a vertex $v$ such that
 $\degr{}{v} = p \geqslant 3$, there are $p$ edges incident to $v$, say
 $\Set{e_1 , \ldots , e_p}$.
 The expansion of $v$ consists in constructing a simple cycle with $p$
 fresh vertices $\Set{v_1 , \ldots , v_p}$ and $p$ fresh edges
 $\Set{e'_1 , \ldots , e'_p}$, and then attaching the original edges
 $e_1 ,\ldots,e_p$ to the cycle thus constructed at the new vertices
 $v_1 , \ldots , v_p$, respectively. An example when $p=4$ is shown in
 Figure~\ref{fig:expand}.
\end{definition}

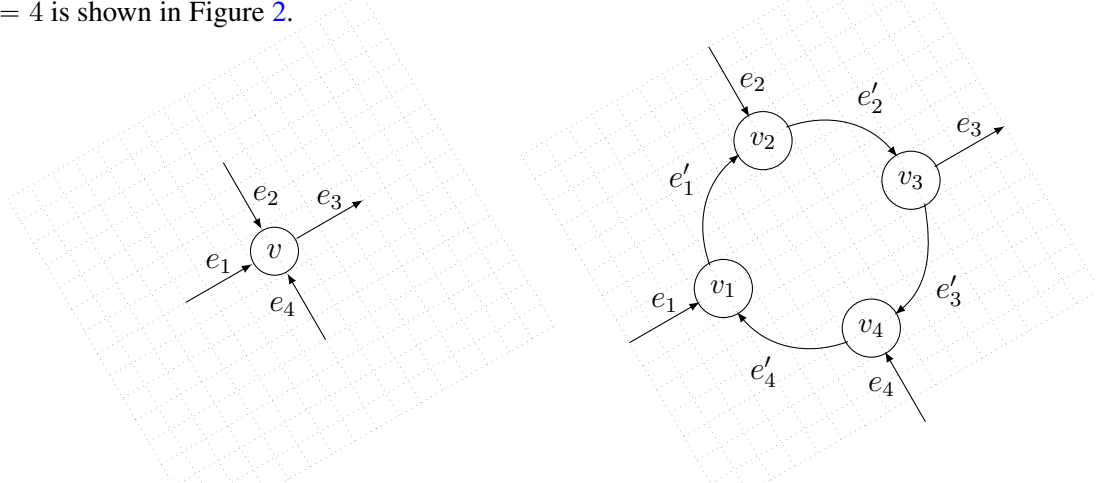
\begin{figure}[h] 
\begin{center}
   \noindent
   \begin{minipage}[b][6cm]{.35\textwidth}

\noindent
%
\rotatebox{30}{\begin{tikzpicture}[scale=.34] 
       \newcommand\EdgeOpacity{[line width=1.4,black,opacity=.98]};
       \draw [help lines, dotted] (0, 0) grid (16,15);
\coordinate (A0) at (8,8);
\coordinate (A1) at (4,8);
\coordinate (B1) at (8,12);
\coordinate (C1) at (12,8);
\coordinate (D1) at (8,4);

\coordinate (A0min) at (7,8);
\coordinate (A0max) at (9,8);
\coordinate (A0up) at (8,9);
\coordinate (A0dn) at (8,7);

\draw (A0) node[rotate=-30,circle,draw] {\large $v$};

\draw[->] (A1) to[out=0,in=180] node[rotate=-30,auto] {\large $e_1$} (A0min);
\draw[->] (B1) to[out=270,in=90] node[rotate=-30,auto] {\large $e_2$} (A0up);
\draw[->] (A0max) to[out=0,in=180] node[rotate=-30,auto] {\large $e_3$} (C1);
\draw[->] (D1) to[out=90,in=270] node[rotate=-30,auto] {\large $e_4$} (A0dn);

\end{tikzpicture}
}

   \end{minipage}
   \qquad 
   \begin{minipage}[b][6cm]{.55\textwidth}

\noindent
%
\rotatebox{30}{\begin{tikzpicture}[scale=.36] 
       \newcommand\EdgeOpacity{[line width=1.4,black,opacity=.98]};
       \draw [help lines, dotted] (0, 0) grid (16,15);
\coordinate (A1) at (0,7);
\coordinate (A2) at (4,7);
\coordinate (A2min) at (3,7);
\coordinate (A2up) at (4,8);
\coordinate (A2dn) at (4,6);
\coordinate (B1) at (8,15);
\coordinate (B2) at (8,11);
\coordinate (B2up) at (8,12);
\coordinate (B2min) at (7,11);
\coordinate (B2max) at (9,11);
\coordinate (C1) at (16,7);
\coordinate (C2) at (12,7);
\coordinate (C2max) at (13,7);
\coordinate (C2up) at (12,8);
\coordinate (C2dn) at (12,6);
\coordinate (D1) at (8,-1);
\coordinate (D2) at (8,3);
\coordinate (D2dn) at (8,2);
\coordinate (D2min) at (7,3);
\coordinate (D2max) at (9,3);

\draw (A2) node[rotate=-30,circle,draw] {$v_1$};
\draw (B2) node[rotate=-30,circle,draw] {$v_2$};
\draw (C2) node[rotate=-30,circle,draw] {$v_3$};
\draw (D2) node[rotate=-30,circle,draw] {$v_4$};

\draw[->] (A2up) to[out=80,in=190] 
          node[rotate=-30,auto] {\large $e'_1$} (B2min);
\draw[->] (A1) to[out=0,in=180] node[rotate=-30,auto] {\large $e_1$} (A2min);
\draw[->] (B2max) to[out=350,in=100]
          node[rotate=-30,auto] {\large $e'_2$} (C2up);
\draw[->] (B1) to[out=270,in=90] node[rotate=-30,auto] {\large $e_2$} (B2up);
\draw[->] (C2dn) to[out=250,in=10]
          node[rotate=-30,auto] {\large $e'_3$} (D2max);
\draw[->] (C2max) to[out=0,in=180] node[rotate=-30,auto] {\large $e_3$} (C1);
\draw[->] (D2min) to[out=170,in=280]
          node[rotate=-30,auto] {\large $e'_4$} (A2dn);
\draw[->] (D1) to[out=90,in=270] node[rotate=-30,auto] {\large $e_4$} (D2dn);

\end{tikzpicture}
}

   \end{minipage}
\end{center}   
\vspace{-.55in}
   \caption{Applying the \textbf{expand} operation
   to a degree-$4$ vertex $v$ (on the
   left) produces a cycle with four new vertices $\Set{v_1,v_2,v_3,v_4}$
   and four new edges $\Set{e'_1,e'_2,e'_3,e'_4}$ (on the right),
   while preserving planarity. 
   }
\label{fig:expand}
\end{figure}

The transformation $G\mapsto \transA{G}$ has two stages
in sequence. Stage 1 eliminates all vertices $v$ such that
$\degr{}{v}\geqslant 4$, and Stage 2 eliminates all two-edges cycles.
\begin{description}[itemsep=1pt,parsep=2pt,topsep=2pt,partopsep=0pt]
\item[Stage 1:] All vertices $v$ such that
     $\degr{}{v}\geqslant 4$ are eliminated by applying the
     \textbf{expand} operation repeatedly, until it cannot
     be applied.
\end{description}
After Stage 1 there are only degree-$3$ vertices in the
transformed directed graph. But we still want to eliminate every
two-edge cycle, \ie, two edges of the form $\set{v\,w}$ and $\set{w\,v}$
where $v\neq w$; we want to eliminate such a two-edge
cycle because $\set{v\,w}$ and $\set{w\,v}$ collapse into a single edge
$\Set{v,w}$ in the undirected version of the graph. This is the
purpose of Stage 2, to prevent such a collapse.
\begin{description}[itemsep=1pt,parsep=2pt,topsep=2pt,partopsep=0pt]     
\item[Stage 2:]
     Every two-edge cycle $\Set{\set{v\,w},\set{w\,v}}$ where
     $v\neq w$ is eliminated by applying the \textbf{expand} operation
     twice, once to each of its endpoints $v$ and $w$, where necessarily
     $\degr{}{v} = \degr{}{w} = 3$ after Stage 1.
\end{description}

Stage 1 and Stage 2 complete the transformation $G\mapsto \transA{G}$.
In words, we have transformed the original $G$ into a $3$-regular $\transA{G}$
by adding ``a few'' directed edges to the former.

\begin{lemma}
  \label{lem:basic}
  Let $G$ be a directed graph satisfying~\nameref{assumption},
  where $\size{\edges{G}} = m$ and $\size{\vertices{G}} = n$. 
  We have the following facts:
\begin{enumerate}[itemsep=1pt,parsep=2pt,topsep=2pt,partopsep=0pt]     
  \item The transformation $G\mapsto \transA{G}$ is carried out in
        linear time $\bigOO{n}$.
  \item $\size{\edges{\transA{G}}} \leqslant 3m$
  \item $\size{\vertices{\transA{G}}} \leqslant n + 2 m$.        
\end{enumerate}
\end{lemma}

\begin{proof}
  The proof of part 1 of the lemma is straightforward, with Stage 1
  and Stage 2 each requiring $\bigOO{n}$ time to do its work. Each of
  the two stages needs to visit each vertex $v$ only once, to test
  whether $v$ satisfies the condition calling for a local
  transformation at $v$ and costing $\bigOO{1}$ time.
  
  For the proof of part 2 of the lemma, note that Stage 1 works on
  vertices $v$ such that $\degr{}{v} \geqslant 4$ such that none of
  the new edges it introduces are involved in two-edge cycles; that
  is, every two-edge cycle that is present \emph{after} Stage 1 is a
  two-edge cycle that is already present \emph{before} Stage 1.  Stage
  2 works on degree-$3$ vertices that are endpoints of two-edge
  cycles, none of them introduced in the Stage 1.

  Let $q$ be the number of edges $e = \set{v\,w}$ or $e = \set{w\,v}$
  with one or two endpoints satisfying one of two conditions:
  \begin{itemize}[itemsep=1pt,parsep=2pt,topsep=2pt,partopsep=0pt]
    \item $\degr{}{v} \geqslant 4$, or
    \item $\degr{}{v} = 3$ and $v$ is one of two vertices on a two-edge cycle; 
  \end{itemize}
  these are the endpoints/vertices worked on during Stage 1 and Stage 2.
  Each edge $e$ of these $q$ edges is associated with one or two new edges,
  depending on whether one or two of $e$'s endpoints are expanded. We conclude:
\[
   \size{\edges{\transA{G}}} \leqslant m + 2q \leqslant m + 2m = 3 m.
\]
  For the proof of part 3 of the lemma, we use the same reasoning as for
  part 2, to show that:
\[
   \size{\vertices{\transA{G}}} \leqslant n + 2q 
   \leqslant n + 2 m .
\]
  We omit the straightforward details.%
  \footnote{\label{foot:first}
  The upper bound $3m$ on $\size{\edges{\transA{G}}}$ is tight,
  in that there are directed graphs $G$
  satisfying~\nameref{assumption} on page~\pageref{assumption}
  such that $\size{\edges{\transA{G}}} = 3m$; this happens
  when the two endpoints of every edge in $G$ are expanded in Stage 1 or
  Stage 2. However, the upper bound $n+2m$ on $\size{\vertices{\transA{G}}}$
  is not tight; this is so because, if vertex $v$ of degree $= p$ is expanded,
  then each of the $p$ incident edges $\Set{e_1,\ldots,e_p}$ contributes one
  new vertex on the cycle replacing $v$, but $v$ itself has to be removed
  from the total count of vertices.}
\end{proof}

The next lemma specializes Lemma~\ref{lem:basic} to the case
of plane directed graphs. It makes clear that for plane directed graphs,
the transformation $G\mapsto\transA{G}$ produces a (small) linear
growth in the size.

\begin{lemma}
  \label{lem:basic-for-plane}
  If $G$ is a plane directed graph satisfying~\nameref{assumption},
  with $\size{\edges{G}} = m$ and $\size{\vertices{G}} = n$ where
  $n\geqslant 3$, then:%
     \footnote{Again here, the upper bounds are not tight. See
     footnote~\ref{foot:first}. But they are easy to compute and
     good enough for our main result.}
  \begin{enumerate}[itemsep=1pt,parsep=2pt,topsep=2pt,partopsep=0pt] 
  \item $\size{\edges{\transA{G}}} \leqslant 18 n - 36$,
  \item $\size{\vertices{\transA{G}}} \leqslant 13 n - 24$, and
  \item $\transA{G}$ is a plane directed graph
         satisfying~\nameref{assumption} such that
     \begin{itemize}[itemsep=1pt,parsep=2pt,topsep=2pt,partopsep=0pt]
        \item[\emph{(3.a)}]
               there are no two-edge cycles in $\transA{G}$,
        \item[\emph{(3.b)}]
               $\degr{}{v} = 3$ for every $v\in\vertices{\transA{G}}$, and
        \item[\emph{(3.c)}]
               $\OutPlan{E}{G} = \OutPlan{E}{\transA{G}}$.
     \end{itemize}
  \end{enumerate}
\end{lemma}

\begin{proof}
Euler's formula (Theorem 4.2.7 and its corollaries
in~\cite{diestel2012}) is usually proved for undirected plane graphs
(no self-loops, no multi edges) and written as $m \leqslant 3n-6$ when
$n\geqslant 3$. But our $G$ is a \emph{directed} plane graph, which
may contain two-edge cycles (but no self-loops). If every double-edge
cycle in $G$ is collapsed into a single edge, we can write
$m/2 \leqslant 3n - 6$, because there are at least $m/2$ edges in
$\undirect{G}$. Hence, $m \leqslant 6n - 12$. Hence also, by parts 2
and 3 in Lemma~\ref{lem:basic}, we have:
\begin{alignat*}{5}
   & \size{\edges{\transA{G}}}\ && \leqslant
     \ 3m \leqslant\ 3(6n-12)\ =\ 18n - 36,
\\[1.12ex]
   & \size{\vertices{\transA{G}}} && \leqslant
     \ n+2m\ \leqslant\ n + 2(6n-12)\ =\ 13n - 24,
\end{alignat*}
as claimed for parts 1 and 2 of the lemma.

For part 3, first note that the transformation
$G\mapsto \transA{G}$ is defined to guarantee (3.a) and (3.b).
Morever, it is readily checked that planarity is an invariant
of every step of the transformation: If $G$ is a plane graph
(not just planar), then so is $\transA{G}$.
Finally, it is readily checked that the equality:
\[
    \OutPlan{E}{G} = \OutPlan{E}{\transA{G}}
\]
is also an invariant of every step of the
transformation $G\mapsto \transA{G}$. The desired conclusion follows.
\end{proof}

We need one more easy lemma.  Let $(G,c,s,t)$ be a flow network.  We
define a new flow network $(\transA{G},\transA{c},\transA{s},\transA{t})$.
The transformation $G\mapsto \transA{G}$ is already defined.  We still
have to define $\transA{c}$, $\transA{s}$, and $\transA{t}$.  In the
transformation $G\mapsto\transA{G}$, every edge $G$ is preserved in
$\transA{G}$, which allows us to view
$\edges{G} \subseteq \edges{\transA{G}}$. So we define:
\[
     \transA{c}(e) :=
     \begin{cases}
     c(e)\quad & \text{if $e\in\edges{G}$},
     \\[1.1ex]
     \text{`a very large capacity'} \quad &
          \text{if $e\in\edges{\transA{G}} - \edges{G}$}.
     \end{cases}
\]
The idea of assigning `a very large capacity' to every new edge introduced
in the transformation $G\mapsto\transA{G}$ is to make these new edges
have no effect in restricting the flow through the network.

If the source $s$ was not expanded into a cycle in the transformation
$G\mapsto\transA{G}$, then $\transA{s} {:=} s$, else $\transA{s} {:=} $ any
of the new vertices on the cycle that replaces $s$. And similarly for the
definition of $\transA{t}$ from the original sink $t$.

Two flow networks $(G_1,c_1,s_1,t_1)$ and
$(G_2,c_2,s_2,t_2)$ are \emph{equivalent} iff
for every flow $f_i : \edges{G_i}\to\nreals$ there is a flow
$f_j : \edges{G_j}\to\nreals$ such that
$\size{f_i} = \size{f_j}$ for all $\Set{i,j} = \Set{1,2}$.

\begin{lemma}
\label{lem:equivalent-networks}
  Let $(G,c,s,t)$ be a flow network, where $G$ is a plane directed
  graph $G$ satisfying~\nameref{assumption} 
  and $\size{\vertices{G}} = n$, and consider the derived flow network
  $(\transA{G},\transA{c},\transA{s},\transA{t})$ as defined above.
  It then holds that:
  \begin{enumerate}[itemsep=1pt,parsep=2pt,topsep=2pt,partopsep=0pt]     
  \item The transformation
        $(G,c,s,t)\mapsto (\transA{G},\transA{c},\transA{s},\transA{t})$
        is carried out in linear time $\bigOO{n}$.
  \item  $(G,c,s,t)$ and $(\transA{G},\transA{c},\transA{s},\transA{t})$
        are equivalent flow networks. %
\Hide{        %
        \footnote{, \ie, \\ for every flow
        $f : \edges{G}\to\reals$ there is a flow
        $\transA{f} : \edges{\transA{G}}\to\reals$ such that
        $\size{f} = \size{\transA{f}}$ and \\
        for every flow $\transA{f} : \edges{\transA{G}}\to\reals$
        there is a flow $f : \edges{G}\to\reals$ such that
        $\size{f} = \size{\transA{f}}$. }
        }
  \end{enumerate}        
\end{lemma}

\begin{proof}
The transformation $G\mapsto \transA{G}$ takes time $\bigOO{n}$, by
part 1 of Lemma~\ref{lem:basic}. The updating from $c$ to $\transA{c}$
takes time $\bigOO{m}$, where $\size{\edges{G}} = m$, and therefore
time $\bigOO{n}$ by Euler's formula (as in the proof of
Lemma~\ref{lem:basic-for-plane}). And setting $\transA{s}$ and $\transA{t}$
takes time $\bigOO{1}$. The conclusion of part 1 follows.

The proof of part 2 is straightforward, since
$\edges{G} \subseteq \edges{\transA{G}}$, with the edges in $G$ preserving
their capacities in $\transA{G}$ and the edges not in $G$ assigned each
`a very large capacity'. All formal details omitted.
\end{proof}

Note that part 2 in Lemma~\ref{lem:equivalent-networks} holds even if
$G$ is not a plane graph, but we do not need this fact for our main
result.  That $G$ is a plane graph is only used in the proof of part 1
in Lemma~\ref{lem:equivalent-networks} to change the complexity bound
from $\bigOO{m+n}$ to $\bigOO{n}$.


\section{Two Previous Results}
\label{sect:two-previous}

  The first result below (Theorem~\ref{thm:about-reassembling}) is
  about the \emph{reassembling problem}, which was studied in earlier
  reports and is here stated in terms of simple undirected graphs (no
  multi-edges, no self loops), but which applies equally well to
  directed graphs satisfying~\nameref{assumption} on
  page~\pageref{assumption}. 

  \paragraph{Graph Reassembling.}
  The \emph{reassembling} of a simple undirected graph $G$ is an abstraction
  of a problem arising in studies of network
  analysis~\cite{BestKfoury:dsl11,Kfoury:sblp11,Kfoury:SCP2014,%
  SouleBestKfouryLapets:eoolt11}.
  There are several equivalent definitions of graph reassembling.  An
  informal intuitive definition was already given in
  Section~\ref{sect:intro}. A formal definition consists in
  constructing a rooted binary tree $\B$ whose nodes are subsets of
  $\vertices{G}$ and whose leaf nodes are singleton sets, with each of
  the latter containing a distinct vertex of $G$.  The parent of two
  nodes in $\B$ is the union of the two children's vertex sets. The
  root node of $\B$ is the full set $\vertices{G}$. If $n
  = \size{\vertices{G}}$, there are thus $n$ leaf nodes in $\B$ and a
  total of $(2n -1)$ nodes in $\B$. We denote the reassembling of $G$
  according to $\B$ by writing $(G,\B)$.%
  \footnote{To keep apart $\B$ and $G$, we reserve the words `node'
  and `branch' for the tree $\B$, and the words `vertex' and `edge'
  for the graph $G$.}

  The \emph{edge-boundary degree} of a node in $\B$ is the number of
  edges that connect vertices in the node's set to vertices not in the
  node's set. Following a terminology used in earlier reports,
  the \emph{$\alpha$-measure} of the reassembling $(G,\B)$, denoted
  $\alpha(G,\B)$, is the largest edge-boundary degree of any node in
  the tree $\B$. We say $\alpha(G,\B)$ is \emph{optimal}
  if it is minimum among all $\alpha$-measures of $G$'s reassemblings,
  in which case we also say $\B$ is \emph{$\alpha$-optimal}.

  The problem of constructing an $\alpha$-optimal reassembling
  $(G,\B)$ of a simple undirected graph $G$ in general was already
  shown NP-hard~\cite[among
  others]{kfoury+mirzaei:2017,kfoury+mirzaei:2017B}.  However,
  restricting attention to \emph{plane} graphs, we have the following
  positive result.

\begin{theorem}
\label{thm:about-reassembling}
  There is an algorithm which, given a plane $3$-regular simple
  undirected graph $G$ as input, returns a reassembling $(G,\B)$ in
  time $\bigOO{n}$ such that $\alpha(G,\B) \leqslant 2 k$, where
  $k = \OutPlan{E}{G}$ and $n = \size{\vertices{G}}$.
\end{theorem}

The value of $\alpha(G,\B)$ returned by the algorithm in
Theorem~\ref{thm:about-reassembling} is independent of $n$; more
precisely, for a fixed $k = \OutPlan{E}{G}$, the value of $n$ can be
arbitrarily large. Note that the algorithm in the theorem only returns
an upper bound $2k$ on $\alpha(G,\B)$ and does not claim that
$\alpha(G,\B)$ is optimal.

Theorem~\ref{thm:about-reassembling} and its proof are in the
report~\cite{kfoury+sisson2018}, which also discusses conditions under
which the bound $2k$ is optimal; specifically, it defines families of
plane $3$-regular simple graphs such that, for any graph $G$ in these
families, $2k$ is the value of an optimal $\alpha(G,\B)$. We do not
use the latter fact in this paper.

\bigskip 

The second result below (Theorem~\ref{thm:about-typing}) is about flow
networks and what are called \emph{network typings}. It is better
stated in terms of what we here call \emph{extended flow networks},
which have an upper bound function on edges $\upperB$, a lower bound
function on edges $\lowerB$, a set of source vertices $S$, and a set
of sink vertices $T$.

 \paragraph{Extended Flow Networks and their Typings.}\label{def:extended-network}
 An \emph{extend flow network} is denoted by a quintuple of the form
 $(G,\upperB,\lowerB,S,T)$ where $G$ is a directed graph
 satisfying~\nameref{assumption} on page~\pageref{assumption} and:
\begin{itemize}[itemsep=1pt,parsep=2pt,topsep=2pt,partopsep=0pt]
  \item $\upperB : \edges{G}\to\nreals$ and $\lowerB : \edges{G}\to\nreals$,
        with $0\leqslant \lowerB(e) \leqslant \upperB(e)$ for every
        $e\in \edges{G}$, and
  \item $\varnothing \neq S\subseteq\vertices{G}$ and
        $\varnothing \neq T\subseteq\vertices{G}$, with
        $S\cap T = \varnothing$.
\end{itemize}
As usual, a \emph{flow} in the network is a function
$f:\edges{G}\to\nreals$. A flow $f$ is \emph{feasible} iff
$\lowerB(e)\leqslant f(e) \leqslant\upperB(e)$
for every $e\in\edges{G}$ and $f$ satisfies flow conservation at every
vertex $v\in \vertices{G} - (S\cup T)$.

An \emph{input-output assignment} (or an \emph{IO assignment}) for
such a network is a function $g:S\cup T\to\nreals$, which expresses the
\emph{excess flow} entering $S$ and exiting $T$.
A \emph{typing} for such a network is a map $\tau$ such that:
\begin{alignat*}{5}
   & \tau : \power{S\cup T}\to\intervals{\reals}\quad\text{where\ }
   \\[1ex]
   & \power{S\cup T} := \SET{\, A \;\big|\; A\subseteq S\cup T\,}
     \text{\ \ and\ \ } \intervals{\reals} :=
     \SET{\,[r_1,r_2]\;\big|\; r_1,r_2\in\reals\text{ and } r_1\leqslant r_2\,}, 
\end{alignat*}
\ie, $\intervals{\reals}$ is the set of bounded closed intervals of reals;
such a typing must satisfy certain soundness conditions (not spelled out here).
An \emph{IO assignment $g$ satisfies the typing $\tau$} iff for every
$A\in\power{S\cup T}$:
\[
      \Big(\sum g (A\cap S) - \sum g (A\cap T)\Big)\ \in\ \tau(A)
\]
where $\sum g (X)$ means $\sum \Set{ g(x)\,|\,x\in X }$.%
   \footnote{By convention, $\sum \varnothing = 0$.}
In particular, if $A = S\cup T$ and $\tau(A) = [r_1,r_2]$, then:
\[
   r_1\ \leqslant\ \sum g(S) - \sum g(T)\ \leqslant\ r_2  .
\]
Hence, one condition for the soundness of the
typing $\tau$ is that we must have $r_1 = r_2 = 0$ when $A = S\cup T$,
\ie, $\tau(S\cup T) = [0,0] = \Set{0}$, 
expressing the fact that the flow entering the network must
equal the flow exiting it.

Given a flow $f:\edges{G}\to\nreals$,
it induces an IO assignment $\induce{f} : S\cup T\to\nreals$ as follows:
\begin{alignat*}{5}
   & \text{for every $s\in S$,}
   \\
   & \induce{f}(s) :=
     \sum\SET{\,f(e)\;\big|\;e=\set{s\,v}\text{ for some } v\in\vertices{G}\,}
     -
     \sum\SET{\,f(e)\;\big|\;e=\set{v\,s}\text{ for some } v\in\vertices{G}\,},
   \\
   & \text{for every $t\in T$,}
   \\
   & \induce{f}(t) :=
     \sum\SET{\,f(e)\;\big|\;e=\set{v\,t}\text{ for some } v\in\vertices{G}\,}
     -
     \sum\SET{\,f(e)\;\big|\;e=\set{t\,v}\text{ for some } v\in\vertices{G}\,}.
\end{alignat*}
\ie, $\induce{f}(s)$ is the total excess flow entering the source $s$ and
$\induce{f}(t)$ is the total excess flow exiting the sink $t$. Thus,
$\sum\induce{f}(S)$ and $\sum\induce{f}(T)$ are the total flows entering
and exiting the network.

As noted in the opening paragraph of this section,
a \emph{reassembling} $\B$ can be defined equally well for a
\emph{directed} graph $G$ satisfying~\nameref{assumption} and
containing no two-edge cycles. This allows us to use $(G,\B)$ and its
measure $\alpha(G,\B)$ in the statement of the next theorem.

\begin{theorem}
\label{thm:about-typing}
If $(G,\upperB,\lowerB,S,T)$ is an extended flow network as defined above
and $(G,\B)$ is a reassembling of the underlying $G$, then we can compute in time
$m\cdot 2^{\bigOO{\delta}}$ a typing $\tau : \power{S\cup T}\to\intervals{\reals}$,
where $m = \size{\edges{G}}$ and
$\delta = \max\,\SET{\,\alpha(G,\B),\,\size{S\cup T}\,}$, such that:
\begin{enumerate}[itemsep=1pt,parsep=2pt,topsep=2pt,partopsep=0pt] 
  \item If $f:\edges{G}\to\nreals$ is a feasible flow, then
        $\induce{f} : S\cup T\to\nreals$ satisfies $\tau$.
  \item If $g : S\cup T\to\nreals$ 
        satisfies $\tau$, then there is a feasible flow
        $f:\edges{G}\to\nreals$ such that $\induce{f} = g$.
\end{enumerate}
In particular, the typing $\tau$ is such that $\tau(S) = [r_1,r_2]$
and $\tau(T) = [-r_2,-r_1]$ for some $r_1,r_2\in\nreals$, with $r_1$ and
$r_2$ being, respectively, the minimum value and the maximum value of
feasible flows in the network.
\end{theorem}

Theorem~\ref{thm:about-typing} and its proof are in the
report~\cite[Theorem 4 on pp. 7-8]{kfoury2018},
which examines other aspects of network typings and their applications.%
   \footnote{There are minor differences between the terminology in this
   paper and the terminology in the report~\cite{kfoury2018}. What is called
   a \emph{binding schedule} $\sigma$ of a graph $G$ and its $\Index{\sigma}$
   in that report are here a \emph{reassembling} $(G,\B)$
   and its measure $\alpha(G,\B)$. }

For a simpler presentation of our main result
(Theorem~\ref{thm:our-result} below), we use
Theorem~\ref{thm:about-typing} with the following restrictions:
$S = \Set{s}$ and $T = \Set{t}$ are singleton sets, and the lower
bound $\lowerB(e) = 0$ for every $e\in \edges{G}$. With these
restrictions, the definition of a network as a quintuple
$(G,\upperB,\lowerB,S,T)$ in Theorem~\ref{thm:about-typing}
matches the definition of a network as a quadruple in
Section~\ref{sect:preliminaries}.  But these restrictions can be
lifted and our result 
re-stated in a more general setting, as in
Theorem~\ref{thm:our-result-extended} below.


\section{The Main Result}
\label{sect:our-result}

We first state and prove the result which is this paper's title, and
then explain how it generalizes to \emph{extended flow
networks} as defined in Sectione~\ref{sect:our-result}. The
time complexity in Theorem~\ref{thm:our-result}
can be written as $\bigOO{n\cdot f(k)}$ where $k$ is an edge-outerplanarity,
$n$ a number of vertices, and $f(k)$ a function of $k$
independent of $n$ -- which thus makes the algorithm in
Theorem~\ref{thm:our-result} `fixed-parameter linear-time' where $k$
is the parameter to keep fixed.

\Hide{
Given a flow network $(G,c,s,t)$ as input where $G$ is plane, our
algorithm returns the value of a max flow in time $f(k)\cdot \bigOO{n}$, where
$k = \OutPlan{E}{G}$ and $n = \size{\vertices{G}}$, for some function $f$.
}

\begin{theorem}
\label{thm:our-result}
There is a fixed-parameter linear-time algorithm to compute the value of
a max flow in plane flow networks $(G,c,s,t)$ where the parameter bound
not to be exceeded is $k=\OutPlan{E}{G}$.
\end{theorem}

\begin{proof}
We can assume the underlying graph $G$ satisfies~\nameref{assumption}
on page~\pageref{assumption}. First, we carry out the transformation
$(G,c,s,t)\mapsto (\transA{G},\transA{c},\transA{s},\transA{t})$
in time $\bigOO{n}$ where $n = \size{\vertices{G}}$,
as described in Lemma~\ref{lem:equivalent-networks},
also according to which $(G,c,s,t)$ and
$(\transA{G},\transA{c},\transA{s},\transA{t})$ are equivalent networks.
According to Lemma~\ref{lem:basic-for-plane}, we have
$\OutPlan{E}{G} = \OutPlan{E}{\transA{G}} = k$ as well as
$\size{\vertices{\transA{G}}} = \bigOO{n}$ and
$\size{\edges{\transA{G}}} = \bigOO{n}$.

To obtain the stated result, it now suffices to apply
Theorems~\ref{thm:about-reassembling} and~\ref{thm:about-typing} to
the transformed network
$(\transA{G},\transA{c},\transA{s},\transA{t})$. In time $\bigOO{n}$,
we first compute a reassembling $\B$ of $\transA{G}$ such that
$\delta = \alpha(\transA{G},\B) \leqslant 2k$, and then compute a typing
$\tau : \power{\Set{\transA{s},\transA{t}}}\to\intervals{\reals}$ in
time $m\cdot 2^{\bigOO{\delta}}$ where $m = \size{\edges{\transA{G}}}= \bigOO{n}$.
If $\tau(\Set{\transA{s}})= [0,r]$ for
some $r\in\reals$, then $r$ is the value of a max flow. The claimed time
complexity follows.
\end{proof}

\paragraph{Remark.}\label{rem:value-of-max-flow}
It is important to note that what is returned by the
algorithm in Theorem~\ref{thm:our-result} is the \emph{value} $r$ of a
max flow, not a \emph{particular} max flow $f:\edges{G}\to\nreals$
such that $\size{f} = r$. It is an additional problem, not considered
in this paper but worthy of study, to compute a particular max flow
$f:\edges{G}\to\nreals$ given that its value $\size{f}$ must be
$r$. While the value $r$ is unique, there are generally many max flows
$f$ such that $\size{f} = r$.

The next result implies the preceding Theorem~\ref{thm:our-result} and
illustrates the flexibility of our method.
Theorem~\ref{thm:our-result-extended} is about 
\emph{extended flow networks}, each of the form $(G,\upperB,\lowerB,S,T)$
where the graph $G$ is a plane directed graph
satisfying~\nameref{assumption} and the extra assumption that
$\size{S\cup T} = \bigOO{k}$ where $k = \OutPlan{E}{G}$. 
A typing $\tau : \power{S\cup T}\to\intervals{\reals}$ for such a
network includes an interval for each $A\in\power{S\cup T}$;
with the extra assumption, the typing has size $2^{\bigOO{k}}$. We impose
the extra assumption in order to keep the complexity linear in
$n = \bigOO{\size{\vertices{G}}}$, though exponential in the parameter
$k = \OutPlan{E}{G}$. 

\begin{theorem}
\label{thm:our-result-extended}

There is a fixed-parameter linear-time algorithm which, given
a plane extended flow network $(G,\upperB,\lowerB,S,T)$ as described in
the preceding paragraph, computes for every $A\in\power{S\cup T}$
a bounded closed interval $[r_1,r_2]$ of reals such that for every feasible
flow $f : \edges{G}\to\nreals$ it holds that:
\[
      r_1\ \leqslant\ \sum \induce{f}(A\cap S) - \sum \induce{f}(A\cap T)
      \ \leqslant\ r_2 .
\]
In particular, if $A = S$, then $r_2$ is the value of a max flow in the extended
network, which is simultaneously returned with the value $r_1$ of a min flow at
no extra cost. The fixed parameter not to be exceeded for the algorithm to work
as claimed is $k = \OutPlan{E}{G}$. 
\end{theorem}

\begin{sketch}
This is a minor variation on the proof of Theorem~\ref{thm:our-result}.
The algorithm starts with the transformation
$(G,\upperB,\lowerB,S,T)\mapsto (\transA{G},\transA{\upperB},\transA{\lowerB},
\transA{S},\transA{T})$ in time $\bigOO{n}$, which is carried out just like the
transformation $(G,c,s,t)\mapsto (\transA{G},\transA{c},\transA{s},\transA{t})$.
One subtle point here: For every new edge $e$
introduced in the transformation $G\mapsto\transA{G}$, we make
$\transA{\lowerB}(e) := 0$
just as we make $\transA{\upperB}(e) := \text{`a very large capacity'}$,
in this way the capacities on the new edges have no effect in resticting
the flow in the transformed network.
The rest of the proof proceeds like the proof of Theorem~\ref{thm:our-result}.
Details omitted. 
\end{sketch}

The same~\nameref{rem:value-of-max-flow} after
Theorem~\ref{thm:our-result} applies to Theorem~\ref{thm:our-result-extended}: 
What is returned by the algorithm are the \emph{values} $r_1$ of a min
flow and $r_2$ of a max flow, not a \emph{particular} min flow
$g:\edges{G}\to\nreals$ and not a \emph{particular} max flow
$f:\edges{G}\to\nreals$ such that $\size{g} = r_1$ and $\size{f} = r_2$.

Compare our result in
Theorem~\ref{thm:our-result-extended} with the main result
in~\cite{borradaile:2011}, where it is shown that
\emph{there exists an algorithm that solves the max-flow problem with
      multiple sources and multiple sinks in an $n$-vertex directed plane
      graph in $\bigOO{n {\log}^3 n}$ time} (with only upper bounds,
      no lower bounds, on edge capacities).


\section{Future Work}
\label{sect:future}

The method proposed in this paper for computing the value of a
maximum flow in planar networks, in fixed-parameter linear time,
can be extended to other more general forms of
flows in planar networks without much trouble, where the parameter bound
not to be exceeded is again edge-outerplanarity.
Under preparation are the four following extensions:
\begin{itemize}[itemsep=1pt,parsep=2pt,topsep=2pt,partopsep=0pt]  
\item
  \emph{multicommodity flows}
  (formal definitions in~\cite[Chapt. 17]{1993orlin}),
\item
  \emph{minimum-cost flows}, \emph{minimum-cost max flows}, and variations
  (definitions in~\cite[Chapt. 9-11]{1993orlin}),
\item  
  \emph{flows with multiplicative gains and losses}, also called
  \emph{generalized flows} (definitions in~\cite[Chapt. 15]{1993orlin}),
\item  
  \emph{flows with additive gains and losses}
  (definitions in~\cite{brandenburg2011shortest}).  
\end{itemize}
To put the relevance of this work in sharper focus, there is no known
algorithm to compute a max flow in any of these four extensions in
linear time in general; in the case of the fourth extension (flows
with additive gains and losses), the problem is known to be
NP-hard~\cite{brandenburg2011shortest}.

We conclude with an open problem. In
the~\nameref{rem:value-of-max-flow} in Section~\ref{sect:our-result},
we pointed out that our method produces the \emph{value} of a maximum
flow, rather than a particular flow with that value, in contrast to
the many other approaches to the maximum-flow problem in the
extant literature.

\paragraph{Open Problem.}
Let $(G,c,s,t)$ be an arbitrary plane flow network.
We can tackle the problem according to one of two approaches:
\begin{enumerate}[itemsep=1pt,parsep=2pt,topsep=2pt,partopsep=0pt]
\item Let the value $r$ of a max flow in $(G,c,s,t)$ 
  be given already. Can we determine in linear time
  a particular max flow $f : \edges{G}\to\nreals$
  such that $\size{f} = r$?
\end{enumerate}  
Alternatively:
\begin{itemize}[itemsep=1pt,parsep=2pt,topsep=2pt,partopsep=0pt]  
\item[2.]
  How can we extend our proposed method so that it simultaneously
  produces the value $r$ of a max flow in $(G,c,s,t)$ \emph{and} a particular
  max flow $f : \edges{G}\to\nreals$ such that $\size{f} = r$ in linear time?
\end{itemize}
A further qualification on the first approach above is
whether the determination of $f$ in linear time can be carried out
without reference to a fixed bound 
$k = \OutPlan{E}{G}$; if this is possible, it will be a stronger result.
In the second approach, since $r$ and $f$
are to be simultaneously determined, it will be a direct extension of our proposed
method which will therefore make explicit reference to
a fixed bound $k = \OutPlan{E}{G}$ for both $r$ and $f$.


\ifTR
\else
\fi


{\footnotesize 
\addcontentsline{toc}{section}{References} 
\bibliographystyle{plain} 
\bibliography{./fixed-parameter-linear-time}
}

\ifTR
\else
\fi

\end{document}